\documentclass{llncs}

\usepackage{url}

\usepackage{amsmath}
\usepackage{amssymb}

\newcommand{\abs}[1]{\left\lvert#1\right\rvert}

\DeclareMathOperator{\Dom}{dom}
\newcommand{\rest}[2]{#1\!\!\restriction_{#2}}
\newcommand{\reste}[2]{#1\restriction_{#2}}

\newcommand{\N}{\mathbb{N}}
\newcommand{\Q}{\mathbb{Q}}
\newcommand{\R}{\mathbb{R}}
\newcommand{\X}{\{0,1\}^*}

\newcommand{\CS}{\Theta}
\newcommand{\CSb}{\overline{\Theta}}

\newcommand{\noi}{\noindent}

\begin{document}


\begin{center}
{\Large \textbf{\boldmath
  A New Representation of Chaitin $\Omega$ Number Based on Compressible Strings
}}
\end{center}

\vspace{-2mm}

\begin{center}
Kohtaro Tadaki
\end{center}

\vspace{-5mm}

\begin{center}
Research and Development Initiative, Chuo University\\
JST CREST\\
1--13--27 Kasuga, Bunkyo-ku, Tokyo 112-8551, Japan\\
E-mail: tadaki@kc.chuo-u.ac.jp
\end{center}

\vspace{-2mm}

\begin{quotation}
\noi\textbf{Abstract.}
In 1975 Chaitin
introduced his $\Omega$ number
as a concrete example of random real.
The real $\Omega$ is defined
based on the set of all halting inputs
for an optimal prefix-free machine $U$,
which
is a universal decoding algorithm used
to define the notion of program-size complexity.
Chaitin showed $\Omega$ to be random
by discovering the property that
the first $n$ bits of the base-two expansion of $\Omega$ solve
the halting problem of $U$
for all binary inputs of length at most $n$.
In this paper,
we introduce a new representation $\CS$ of Chaitin $\Omega$ number.
The real $\CS$ is defined based on the set of all compressible strings.
We investigate the properties of $\CS$ and show that $\CS$ is random.
In addition,
we generalize $\CS$ to two directions $\CS(T)$ and $\CSb(T)$ with a real $T>0$.
We then study their properties.
In particular,
we show that
the computability of the real $\CS(T)$ gives a sufficient condition
for a real $T\in(0,1)$ to be a fixed point on partial randomness,
i.e., to satisfy the condition that the compression rate of $T$ equals to $T$.
\end{quotation}

\begin{quotation}
\noi\textit{Key words\/}:
algorithmic information theory,
Chaitin $\Omega$ number,
randomness,
partial randomness,
fixed point,
program-size complexity
\end{quotation}

\vspace*{-4mm}

\section{Introduction}

\vspace*{-2mm}

Algorithmic information theory (AIT, for short) is a framework
for applying
information-theoretic and probabilistic ideas to recursive function theory.
One of the primary concepts of AIT is the \textit{program-size complexity}
(or \textit{Kolmogorov complexity}) $H(s)$ of a finite binary string $s$,
which is defined as the length of the shortest binary
input
for a universal decoding algorithm $U$,
called an \textit{optimal prefix-free machine},
to output $s$.
By the definition, $H(s)$ is thought to represent
the amount of randomness contained in a finite binary string $s$,
which cannot be captured in an effective manner.
In particular,
the notion of program-size complexity plays a crucial role in
characterizing the \textit{randomness} of an infinite binary string,
or equivalently, a real.
In \cite{C75} Chaitin introduced the halting probability $\Omega$
as a concrete example of
random real.
His $\Omega$ is defined
based on the set of all halting inputs for $U$,
and plays a central role in the
metamathematical
development of
AIT \cite{C87b}.
The first $n$ bits of the base-two expansion of $\Omega$ solve
the halting problem of $U$ for inputs of length at most $n$.
Based on this property,
Chaitin showed that $\Omega$ is random.

In this paper, we introduce a new representation $\CS$ of Chaitin $\Omega$ number.
The real $\CS$ is defined based on the set of all compressible strings,
i.e., all finite binary strings $s$ such that $H(s)<\abs{s}$,
where $\abs{s}$ is the length of $s$.
The first $n$ bits of the base-two expansion of $\CS$ enables us
to calculate a random finite string of length $n$,
i.e., a finite binary string $s$ for which $\abs{s}=n$ and $\abs{s}\le H(s)$.
Based on this property,
we show that $\CS$ is random.

In the work \cite{T02}
we introduced the notion of \textit{partial randomness} for a real as
a stronger representation of the compression rate of a real
by means of program-size complexity.
At the same time,
we generalized
the
halting probability $\Omega$ to $Z(T)$
with a real $T$
so that,
for every $T\in(0,1]$, if $T$ is computable
then the partial randomness of the real $Z(T)$ exactly equals to $T$.%
\footnote{In
\cite{T02},
$Z(T)$ is denoted by $\Omega^T$.
}
In the case of $T=1$,
$Z(T)$ results in
$\Omega$, i.e., $Z(1)=\Omega$.
Later on,
in the work \cite{T08CiE}
we revealed a special significance of the computability of the value $Z(T)$.
Namely,
we proved \textit{the fixed point theorem on partial randomness},%
\footnote{
The fixed point theorem on partial randomness is called
a fixed point theorem on compression rate in \cite{T08CiE}.}
which states that, for every $T\in(0,1)$,
if $Z(T)$ is a computable real,
then the partial randomness of $T$ equals to $T$,
and therefore the compression rate of $T$ equals to $T$,
i.e.,
$\lim_{n\to\infty}H(\rest{T}{n})/n=T$,
where $\rest{T}{n}$ is the first $n$ bits of the base-two expansion of $T$.

In a similar manner to the generalization of $\Omega$ to $Z(T)$,
in this paper we generalize $\CS$ to two directions $\CS(T)$ and $\CSb(T)$.
We then show that the reals $\CS(T)$ and $\CSb(T)$ both have the same randomness properties
as $Z(T)$.
In particular,
we show that the fixed point theorem on partial randomness,
which has the same form as for $Z(T)$, holds for $\CS(T)$.

%
The paper is organized as follows.
We begin in Section~\ref{preliminaries} with
some preliminaries to AIT and partial randomness.
In Section~\ref{new}
we introduce $\CS$ and study its property.
Subsequently,
we generalize $\CS$ to two directions $\CS(T)$ and $\CSb(T)$
in Section~\ref{gen1} and Section~\ref{gen2}, respectively.
In Section~\ref{fptgen1},
we prove the fixed point theorem on partial randomness based on
the computability of
the value $\CS(T)$.

\vspace*{-3mm}

\section{Preliminaries}
\label{preliminaries}

\vspace*{-2mm}


We start with some notation about numbers and strings
which will be used in this paper.
$\#S$ is the cardinality of $S$ for any set $S$.
$\N=\left\{0,1,2,3,\dotsc\right\}$ is the set of natural numbers,
and $\N^+$ is the set of positive integers.
$\Q$ is the set of rationals, and
$\R$ is the set of reals.
%
A sequence $\{a_n\}_{n\in\N}$ of numbers
(rationals or reals)
is called \textit{increasing} if $a_{n+1}>a_{n}$ for all $n\in\N$.
Normally,
$O(1)$ denotes any function $f\colon \N^+\to\R$ such that
there is $C\in\R$ with the property that
$\abs{f(n)}\le C$ for all $n\in\N^+$.
On the other hand,
$o(n)$ denotes any function $g\colon \N^+\to\R$ such
that $\lim_{n \to \infty}g(n)/n=0$.

$\X=
\left\{
  \lambda,0,1,00,01,10,11,000,\dotsc
\right\}$
is the set of finite binary strings
where $\lambda$ denotes the \textit{empty string},
and $\X$ is ordered as indicated.
We identify any string in $\X$ with a natural number in this order,
i.e.,
we consider $\varphi\colon \X\to\N$ such that $\varphi(s)=1s-1$
where the concatenation $1s$ of strings $1$ and $s$ is regarded
as a dyadic integer,
and then we identify $s$ with $\varphi(s)$.
For any $s \in \X$, $\abs{s}$ is the \textit{length} of $s$.
For any $n\in\N$, we denote by $\{0,1\}^n$
the set $\{\,s\mid s\in\X\;\&\;\abs{s}=n\}$.
A subset $S$ of $\X$ is called
\textit{prefix-free}
if no string in $S$ is a prefix of another string in $S$.
For any
function $f$,
the domain of definition of $f$ is denoted by $\Dom f$.
We write ``r.e.'' instead of ``recursively enumerable.''

Let $\alpha$ be an arbitrary real.
$\lfloor \alpha \rfloor$ is the greatest integer less than or equal to $\alpha$,
and $\lceil \alpha \rceil$ is the smallest integer greater than or equal to $\alpha$.
For any $n\in\N^+$,
we denote by $\rest{\alpha}{n}\in\X$
the first $n$ bits of the base-two expansion of
$\alpha - \lfloor \alpha \rfloor$ with infinitely many zeros.
For example,
in the case of $\alpha=5/8$,
$\rest{\alpha}{6}=101000$.
A real $\alpha$ is called \textit{right-computable} if
there exists a total recursive function $f\colon\N^+\to\Q$ such that
$\alpha\le f(n)$ for all $n\in\N^+$ and $\lim_{n\to\infty} f(n)=\alpha$.
On the other hand,
a real $\alpha$ is called \textit{left-computable} if
$-\alpha$ is right-computable.
A left-computable real is also called an \textit{r.e.}~real.
A real $\alpha$ is called \textit{computable} if
there exists a total recursive function $f\colon\N^+\to\Q$ such that
$\abs{\alpha-f(n)} < 1/n$ for all $n\in\N^+$.
It is then easy to show the following theorem.

\vspace*{-1mm}

\begin{theorem}\label{rfcomputable}
Let $\alpha\in\R$.
\vspace*{-2mm}
\begin{enumerate}
  \item $\alpha$ is computable if and only if
    $\alpha$ is both right-computable and left-computable.
  \item $\alpha$ is right-computable if and only if
    the set $\{\,r\in\Q\mid \alpha<r\,\}$ is r.e.\qed
\end{enumerate}
\end{theorem}

\vspace*{-5mm}

\subsection{Algorithmic Information Theory}
\label{ait}

In the following
we concisely review some definitions and results of
AIT
\cite{C75,C87b}.
A \textit{prefix-free machine} is a partial recursive function
$C\colon \X\to \X$
such that
$\Dom C$ is a prefix-free set.
For each prefix-free machine $C$ and each $s \in \X$,
$H_C(s)$ is defined
by
$H_C(s) =
\min
\left\{\,
  \abs{p}\,\big|\;p \in \X\>\&\>C(p)=s
\,\right\}$
(may be $\infty$).
A prefix-free machine $U$ is said to be \textit{optimal} if
for each prefix-free machine $C$ there exists $d\in\N$
with the following property;
if $p\in\Dom C$, then there is $q\in\Dom U$ for which
$U(q)=C(p)$ and $\abs{q}\le\abs{p}+d$.
It is easy to see that there exists an optimal prefix-free machine.
We choose a particular optimal prefix-free machine $U$
as the standard one for use,
and define $H(s)$ as $H_U(s)$,
which is referred to as
the \textit{program-size complexity} of $s$ or
the \textit{Kolmogorov complexity} of $s$.
It follows that
for every prefix-free machine $C$ there exists $d\in\N$ such that,
for every $s\in\X$,
\vspace*{-1mm}
\begin{equation}\label{minimal}
  H(s)\le H_C(s)+d.
\end{equation}
\vspace*{-4mm}\\
Based on this we can show that,
for every partial recursive function $\Psi\colon \X\to \X$,
there exists $d\in\N$ such that,
for every $s \in \Dom \Psi$,
\vspace*{-2mm}
\begin{equation}\label{Psi}
  H(\Psi(s))\le H(s)+d.
\end{equation}
\vspace*{-6mm}\\
Based on \eqref{minimal} we can also show that
there exists $d\in\N$ such that,
for every $s\neq\lambda$,
\vspace*{-3mm}
\begin{equation}\label{Hlabs}
  H(s)\le \abs{s}+2\log_2 \abs{s}+d.
\end{equation}
\vspace*{-4mm}

For any $s\in\X$,
we define $s^*$ as $\min\{\,p\in\X\mid U(p)=s\}$,
i.e., the first element in the ordered set $\X$
of all strings $p$ such that $U(p)=s$.
Then, $\abs{s^*}=H(s)$ for every $s\in\X$.


Chaitin~\cite{C75} introduced $\Omega$ number as follows.
For each optimal prefix-free machine $V$,
the halting probability $\Omega_V$ of $V$ is defined
by
\vspace*{-1mm}
\begin{equation*}
  \Omega_V=\sum_{p\in\Dom V}2^{-\abs{p}}.
\end{equation*}
\vspace*{-4mm}\\
For every optimal prefix-free machine $V$,
since $\Dom V$ is prefix-free,
$\Omega_V$ converges and $0<\Omega_V\le 1$.
%
%
For any $\alpha\in\R$,
we say that $\alpha$ is \textit{weakly Chaitin random}
if there exists $c\in\N$ such that
$n-c\le H(\rest{\alpha}{n})$ for all $n\in\N^+$
\cite{C75,C87b}.
%
%
Chaitin \cite{C75} showed that
$\Omega_V$ is weakly Chaitin random
for every optimal prefix-free machine $V$.
Therefore $0<\Omega_V<1$ for every optimal prefix-free machine $V$.
%
%


\vspace*{-2mm}

\subsection{Partial Randomness}
\label{partial}

\vspace*{-1mm}

In the work \cite{T02},
we generalized the notion of
the randomness of
a real
so that \textit{the degree of the randomness},
which is often referred to
as
\textit{the partial randomness} recently,
can be characterized by a real $T$
with $0<T\le 1$ as follows.

\begin{definition}[\boldmath weak Chaitin $T$-randomness]
  Let $T\in(0,1]$ and let $\alpha\in\R$.
  We say that $\alpha$ is \textit{weakly Chaitin $T$-random} if
  there exists $c\in\N$ such that, for all $n\in\N^+$,
  $Tn-c \le H(\rest{\alpha}{n})$.
  \qed
\end{definition}

In the case where $T=1$,
the weak Chaitin $T$-randomness results in the weak Chaitin randomness.

%

\begin{definition}[\boldmath $T$-compressibility and strict $T$-compressibility]
Let $T\in(0,1]$ and let $\alpha\in\R$.
We say that $\alpha$ is \textit{$T$-compressible} if
$H(\rest{\alpha}{n})\le Tn+o(n)$,
namely, if
$\limsup_{n \to \infty}H(\rest{\alpha}{n})/n\le T$.
We say that $\alpha$ is \textit{strictly $T$-compressible} if
there exists $d\in\N$ such that, for all $n\in\N^+$,
$H(\rest{\alpha}{n})\le Tn+d$.
\qed
\end{definition}

For every real $\alpha$,
if $\alpha$ is weakly Chaitin $T$-random and $T$-compressible,
then $\lim_{n\to \infty} H(\rest{\alpha}{n})/n = T$,
i.e., the \textit{compression rate} of $\alpha$ equals to $T$.

In the work \cite{T02},
we generalized
Chaitin $\Omega$ number
to $Z(T)$
as follows.
%
For each optimal prefix-free machine $V$ and each real $T>0$,
the \textit{generalized halting probability} $Z_V(T)$ of $V$ is
defined
by
\vspace*{-2mm}
\begin{equation*}
  Z_V(T) = \sum_{p\in\Dom V}2^{-\frac{\abs{p}}{T}}.
\end{equation*}
\vspace*{-4mm}\\
Thus,
$Z_V(1)=\Omega_V$.
If $0<T\le 1$, then $Z_V(T)$ converges and $0<Z_V(T)<1$,
since $Z_V(T)\le \Omega_V<1$.
The following theorem holds for $Z_V(T)$.

\begin{theorem}[Tadaki \cite{T02}]\label{pomgd}
Let $V$ be an optimal prefix-free machine.
\vspace*{-2mm}
\begin{enumerate}
  \item If $0<T\le 1$ and $T$ is computable,
    then $Z_V(T)$ is a left-computable real which
    is weakly Chaitin $T$-random and $T$-compressible.
  \item If $1<T$, then $Z_V(T)$ diverges to $\infty$.\qed
\end{enumerate}
\end{theorem}





The computability of the value $Z_V(T)$ has a special implication on $T$
as follows.

\begin{theorem}[fixed point theorem on partial randomness, Tadaki~\cite{T08CiE}]
\label{fpt}
Let $V$ be an optimal prefix-free machine.
For every $T\in(0,1)$,
if $Z_V(T)$ is computable,
then $T$ is weakly Chaitin $T$-random and $T$-compressible,
and therefore
\vspace*{-2mm}
\begin{equation}\label{fppr}
  \lim_{n\to\infty}\frac{H(\rest{T}{n})}{n}=T.
\end{equation}
\qed
\end{theorem}

The equality \eqref{fppr}
means that the compression rate of $T$ equals to $T$ itself.
Intuitively,
we might interpret the meaning of \eqref{fppr} as follows:
Consider imaginarily a file of infinite size whose content is
\begin{center}
``The compression rate of this file is $0.100111001\dotsc\dotsc$''
\end{center}
When this file is compressed,
the compression rate of this file actually
equals to $0.100111001\dotsc\dotsc$,
as the content of this file says.
This situation is self-referential and forms a fixed point.
For a simple and self-contained proof of Theorem~\ref{fpt},
see Section~5 of Tadaki~\cite{T09FICS}.

A left-computable real has a special property on partial randomness,
as shown in Theorem~\ref{mfcs10} below.

\begin{definition}[\boldmath
$T$-convergence, Tadaki \cite{T09MFCS}]
Let $T\in(0,1]$.
An increasing sequence $\{a_n\}$ of reals is called
\textit{$T$-convergent} if
$\sum_{n=0}^{\infty} (a_{n+1}-a_{n})^T\!<\infty$.
A left-computable real $\alpha$ is called \textit{$T$-convergent} if
there exists a $T$-convergent computable,
increasing sequence of rationals which
converges to $\alpha$.
\qed
\end{definition}

\begin{theorem}[Tadaki \cite{T10MFCS}]\label{mfcs10}
Let $T$ be a computable real with $0<T<1$.
For every left-computable real $\alpha$,
if $\alpha$ is $T$-convergent then $\alpha$ is strictly $T$-compressible.
\qed 
\end{theorem}


\vspace*{-6mm}

\section{New Representation of Chaitin $\Omega$ Number}
\label{new}

In this section,
we introduce a new representation $\CS$ of Chaitin $\Omega$ number
based on the set of all compressible strings,
and investigate its property.

\begin{definition}
For any optimal prefix-free machine $V$,
$\CS_V$ is defined
by
\begin{equation*}
  \CS_V=\sum_{H_V(s)<\abs{s}}2^{-\abs{s}},
\end{equation*}
where the sum is over all $s\in\X$ such that $H_V(s)<\abs{s}$.
\qed
\end{definition}

For each optimal prefix-free machine $V$,
we see that
\begin{equation*}
  \CS_V
  <\sum_{H_V(s)<\abs{s}}2^{-H_V(s)}
  \le\sum_{s\in\X}2^{-H_V(s)}
  \le\sum_{p\in\Dom V}2^{-\abs{p}}
  =\Omega_V.
\end{equation*}
Thus, $\CS_V$ converges and $0<\CS_V<\Omega_V$
for every optimal prefix-free machine $V$.
It is important to evaluate
how many strings $s$ satisfy the condition $H_V(s)<\abs{s}$.
For that purpose,
we define
$S_V(n)=\{\,s\in\X\mid\abs{s}=n\;\&\;H_V(s)<n\,\}$
for each optimal prefix-free machine $V$ and each $n\in\N$.
We can then show the following theorem.

\begin{theorem}\label{evalsHsln}
Let $V$ be an optimal prefix-free machine.
Then $S_V(n)\subsetneqq \{0,1\}^n$ for every $n\in\N$.
Moreover
$\#S_V(n)=2^{n-H(n)+O(1)}$
for all $n\in\N^+$,
i.e.,
there exists $d\in\N$ such that
(i) $\#S_V(n)\le 2^{n-H(n)+d}$ for all $n\in\N$, and
(ii) $2^{n-H(n)-d}\le\#S_V(n)$ for all sufficiently large $n\in\N$.
\qed
\end{theorem}

\vspace*{-1mm}


The first half of Theorem~\ref{evalsHsln} is easily shown
by counting the number of binary strings of length less than $n$.
Solovay~\cite{Sol75} showed that
$\#\{\,s\in\X\mid H_V(s)<n\,\}=2^{n-H(n)+O(1)}$
for every optimal prefix-free machine $V$.
The last half of Theorem~\ref{evalsHsln} slightly improves this result.

\vspace*{-1mm}

\begin{theorem}\label{CSVwCr}
For every optimal prefix-free machine $V$,
$\CS_V$ is a left-computable real which is weakly Chaitin random.
\qed
\end{theorem}

\vspace*{-1mm}

Theorem~\ref{CSVwCr} results from
each of Theorem~\ref{pCSVT} (i) and Theorem~\ref{pCSbVT} (i) below
by setting $T=1$.
Thus,
we here omit the proof of Theorem~\ref{CSVwCr}.

The works of Calude, et al.~\cite{CHKW01} and Ku\v{c}era and Slaman \cite{KS01}
showed that,
for every $\alpha\in(0,1)$,
$\alpha$ is left-computable and weakly Chaitin random if and only if
there exists an optimal prefix-free machine $V$ such that $\alpha=\Omega_V$.
Thus, it follows from Theorem~\ref{CSVwCr} that,
for every optimal prefix-free machine $V$,
there exists an optimal prefix-free machine $W$ such that $\CS_V=\Omega_W$.
However,
it is open whether the following holds or not:
For every optimal prefix-free machine $W$,
there exists an optimal prefix-free machine $V$ such that $\Omega_W=\CS_V$.

In the subsequent two sections,
we generalize $\CS_V$ to two directions $\CS_V(T)$ and $\CSb_V(T)$ with a real $T>0$.
We see that the reals $\CS_V(T)$ and $\CSb_V(T)$ both have
the same randomness properties as $Z_V(T)$
(i.e., the properties shown in Theorem~\ref{pomgd} for $Z_V(T)$).

\vspace*{-3mm}

\section{Generalization of $\CS$ to $\CS(T)$}
\label{gen1}

\vspace*{-2mm}

\begin{definition}
For any optimal prefix-free machine $V$ and any real $T>0$,
$\CS_V(T)$ is defined
by
\begin{equation*}
  \CS_V(T)=\sum_{H_V(s)<\abs{s}}2^{-\frac{\abs{s}}{T}}.
\end{equation*}
\vspace*{-12mm}\\
\qed
\end{definition}
\vspace*{4mm}

Thus,
$\CS_V(1)=\CS_V$.
If $0<T\le 1$, then $\CS_V(T)$ converges and $0<\CS_V(T)<1$,
since $\CS_V(T)\le \CS_V<1$.
The following theorem holds for $\CS_V(T)$.

\begin{theorem}\label{pCSVT}
Let $V$ be an optimal prefix-free machine, and let $T>0$.
\vspace*{-2mm}
\begin{enumerate}
  \item If $T$ is computable and $0<T\le 1$,
    then $\CS_V(T)$ is a left-computable real
    which is weakly Chaitin $T$-random.
  \item If $T$ is computable and $0<T<1$,
    then $\CS_V(T)$ is strictly $T$-compressible.
  \item If $1<T$, then $\CS_V(T)$ diverges to $\infty$.
\end{enumerate}
\end{theorem}

\begin{proof}
Let $V$ be an optimal prefix-free machine.
We first note that,
for every $s\in\X$,
$H_V(s)<\abs{s}$ if and only if
there exists $p\in\Dom V$ such that $V(p)=s$ and $\abs{p}<\abs{s}$.
Thus, the set $\{\,s\in\X\mid H_V(s)<\abs{s}\,\}$ is
r.e.~and, obviously, infinite.
Let $s_1,s_2,s_3, \dotsc$ be
a particular recursive enumeration of this set.

(i) Suppose that $T$ is a computable real and $0<T\le 1$.
Then, since $\CS_V(T)=\sum_{i=1}^{\infty} 2^{-\abs{s_i}/T}$,
it is easy to see that $\CS_V(T)$ is left-computable.

For each $n\in\N^+$,
let $\alpha_n$ be the first $n$ bits of
the base-two expansion of $\CS_V(T)$ with infinitely many ones.
Then,
since $0.\alpha_n<\CS_V(T)$ for every $n\in\N^+$,
$\sum_{i=1}^{\infty}2^{-\abs{s_i}}=\CS_V(T)$,
and $T$ is computable,
there exists a partial recursive function $\xi\colon\X\to \N^+$ such that,
for every $n\in\N^+$,
$0.\alpha_n < \sum_{i=1}^{\xi(\alpha_n)}2^{-\abs{s_i}/T}$.
It is then easy to see that
$\sum_{i=\xi(\alpha_n)+1}^{\infty}2^{-\abs{s_i}/T}
=\CS_V(T)-\sum_{i=1}^{\xi(\alpha_n)}2^{-\abs{s_i}/T}
<\CS_V(T)-0.\alpha_n
< 2^{-n}$
for every $n\in\N^+$.
It follows that,
for all $i>\xi(\alpha_n)$,
$2^{-\abs{s_i}/T}<2^{-n}$ and therefore $Tn<\abs{s_i}$.
Thus, given $\alpha_n$,
by calculating the set
$\{\,s_i\mid i\le \xi(\alpha_n)\;\&\;\abs{s_i}=\lfloor Tn\rfloor\,\}$
and picking any
one
finite binary string of length $\lfloor Tn\rfloor$ which is not in this set,
one can obtain $s\in\{0,1\}^{\lfloor Tn\rfloor}$ such that $\abs{s}\le H_V(s)$.
This is possible since
$\{\,s_i\mid i\le \xi(\alpha_n)\;\&\;\abs{s_i}=\lfloor Tn\rfloor\,\}
=S_V(\lfloor Tn\rfloor)
\subsetneqq \{0,1\}^{\lfloor Tn\rfloor}$,
where the last proper inclusion is due to the first half of Theorem~\ref{evalsHsln}.

Hence, there exists a partial recursive function $\Psi\colon\X\to\X$
such that
$\lfloor Tn\rfloor\le H_V(\Psi(\alpha_n))$.
Using the optimality of $V$,
we then see that $Tn\le H(\Psi(\alpha_n))+O(1)$ for all $n\in\N^+$.
On the other hand,
it follows from \eqref{Psi} that
there exists $c_\Psi\in\N$ such that
$H(\Psi(\alpha_n))\le H(\alpha_n)+c_\Psi$.
Therefore, we have
\vspace*{0mm}
\begin{equation}\label{nHao}
  Tn\le H(\alpha_n)+O(1)
\end{equation}
\vspace*{-5mm}\\
for all $n\in\N^+$.
This inequality implies that $\CS_V(T)$ is not computable and therefore
the base-two expansion of $\CS_V(T)$ with infinitely many ones has
infinitely many zeros also.
Hence $\alpha_n=\rest{\CS_V(T)}{n}$ for every $n\in\N^+$.
It follows from \eqref{nHao} that $\CS_V(T)$ is weakly Chaitin $T$-random.

(ii) Suppose that $T$ is a computable real and $0<T<1$.
Note that $\CS_V(T)=\sum_{i=1}^{\infty} 2^{-\abs{s_i}/T}$ and
$\sum_{i=1}^{\infty} (2^{-\abs{s_i}/T})^T
=\sum_{i=1}^{\infty} 2^{-\abs{s_i}}
=\CS_V
<\infty$.
Thus, since $T$ is computable, it is easy to show that
$\CS_V(T)$ is a $T$-convergent left-computable real.
It follows from Theorem~\ref{mfcs10} that
$\CS_V(T)$ is strictly $T$-compressible.

(iii)
Suppose that $T>1$.
We then choose a particular computable real $t$ satisfying $1<t\le T$.
Let us first assume that $\CS_V(t)$ converges.
Based on an argument similar to the proof of
Theorem \ref{pCSVT} (i),
it is easy to show that $\CS_V(t)$ is weakly Chaitin $t$-random,
i.e., there exists $c\in\N$ such that
$tn-c\le H(\rest{\CS_V(t)}{n})$ for all $n\in\N^+$.
It follows from \eqref{Hlabs} that $tn-c\le n+o(n)$.
Dividing by $n$ and letting $n\to\infty$ we have $t\le 1$,
which contradicts the fact $t>1$.
Thus, $\CS_V(t)$ diverges to $\infty$.
By noting $\CS_V(t)\le\CS_V(T)$ we see that $\CS_V(T)$
diverges to $\infty$.
\qed
\end{proof}

\vspace*{-3mm}

\section{Generalization of $\CS$ to $\CSb(T)$}
\label{gen2}

\vspace*{-1mm}

\begin{definition}
For any optimal prefix-free machine $V$ and any real $T>0$,
$\CSb_V(T)$ is defined
by
\begin{equation*}
  \CSb_V(T)=\sum_{H_V(s)<T\abs{s}}2^{-\abs{s}},
\end{equation*}
where the sum is over all $s\in\X$ such that $H_V(s)<T\abs{s}$.
\qed
\end{definition}

Thus,
$\CSb_V(1)=\CS_V$.
For each optimal prefix-free machine $V$ and each real $T$ with $0<T\le 1$,
we see that
\vspace*{-1mm}
\begin{equation*}
  \CSb_V(T)
  <\sum_{H_V(s)<T\abs{s}}2^{-\frac{H_V(s)}{T}}
  \le\sum_{s\in\X}2^{-\frac{H_V(s)}{T}}
  \le\sum_{p\in\Dom V}2^{-\frac{\abs{p}}{T}}
  =Z_V(T).
\end{equation*}
\vspace*{-5mm}\\
Thus, $\CSb_V(T)$ converges and $0<\CSb_V(T)<Z_V(T)$
for every optimal prefix-free machine $V$
and every real $T$ with $0<T\le 1$.
We define
$S_{V,T}(n)=\{\,s\in\X\mid\abs{s}=n\;\&\;H_V(s)<Tn\,\}$
for each optimal prefix-free machine $V$, each $T\in(0,1]$, and each $n\in\N$.
It follows from Theorem~\ref{evalsHsln} that
$S_{V,T}(n)\subset S_V(n)\subsetneqq \{0,1\}^n$
for every optimal prefix-free machine $V$, every $T\in(0,1]$, and every $n\in\N$.
The following theorem holds for $\CSb_V(T)$.

\begin{theorem}\label{pCSbVT}
Let $V$ be an optimal prefix-free machine, and let $T>0$.
\vspace*{-2mm}
\begin{enumerate}
  \item If $T$ is left-computable and $0<T\le 1$,
    then $\CSb_V(T)$ is a left-computable real
    which is weakly Chaitin $T$-random.
  \item If $T$ is computable and $0<T<1$,
    then $\CSb_V(T)$ is strictly $T$-compressible.
  \item If $1<T$, then $\CSb_V(T)$ diverges to $\infty$.
\end{enumerate}
\end{theorem}

\begin{proof}
Let $V$ be an optimal prefix-free machine.

(i) Suppose that $T$ is a left-computable real and $0<T\le 1$.
We first note that,
for every $s\in\X$,
$H_V(s)<T\abs{s}$ if and only if
there exists $p\in\Dom V$ such that $V(p)=s$ and $\abs{p}<T\abs{s}$.
Thus, since $T$ is left-computable,
the set $\{\,s\in\X\mid H_V(s)<T\abs{s}\,\}$ is
r.e.~and, obviously, infinite.
Let $s_1,s_2,s_3, \dotsc$ be
a particular recursive enumeration of this set.
Then, since $\CSb_V(T)=\sum_{i=1}^{\infty} 2^{-\abs{s_i}}$,
it is easy to see that $\CSb_V(T)$ is left-computable.

For each $n\in\N^+$,
let $\alpha_n$ be the first $n$ bits of
the base-two expansion of $\CSb_V(T)$ with infinitely many ones.
Then,
since $0.\alpha_n<\CSb_V(T)$ for every $n\in\N^+$ and
$\sum_{i=1}^{\infty}2^{-\abs{s_i}}=\CSb_V(T)$,
there exists a partial recursive function $\xi\colon\X\to \N^+$ such that,
for every $n\in\N^+$,
$0.\alpha_n < \sum_{i=1}^{\xi(\alpha_n)}2^{-\abs{s_i}}$.
It is then easy to see that
$\sum_{i=\xi(\alpha_n)+1}^{\infty}2^{-\abs{s_i}}
=\CSb_V(T)-\sum_{i=1}^{\xi(\alpha_n)}2^{-\abs{s_i}}
<\CSb_V(T)-0.\alpha_n
< 2^{-n}$
for every $n\in\N^+$.
It follows that,
for all $i>\xi(\alpha_n)$,
$2^{-\abs{s_i}}<2^{-n}$ and therefore $n<\abs{s_i}$.
Thus, given $\alpha_n$,
by calculating the set
$\{\,s_i\mid i\le \xi(\alpha_n)\;\&\;\abs{s_i}=n,\}$
and picking any
one
finite binary string of length $n$ which is not in this set,
one can obtain $s\in\{0,1\}^{n}$ such that $T\abs{s}\le H_V(s)$.
This is possible since
$\{\,s_i\mid i\le \xi(\alpha_n)\;\&\;\abs{s_i}=n\,\}
=S_{V,T}(n)
\subsetneqq \{0,1\}^{n}$.

Hence, there exists a partial recursive function $\Psi\colon\X\to\X$
such that
$Tn\le H_V(\Psi(\alpha_n))$.
Using the optimality of $V$,
we then see that $Tn\le H(\Psi(\alpha_n))+O(1)$ for all $n\in\N^+$.
On the other hand,
it follows from \eqref{Psi} that
there exists $c_\Psi\in\N$ such that
$H(\Psi(\alpha_n))\le H(\alpha_n)+c_\Psi$.
Therefore, we have
\begin{equation}\label{nHaob}
  Tn\le H(\alpha_n)+O(1)
\end{equation}
for all $n\in\N^+$.
This inequality implies that $\CSb_V(T)$ is not computable and therefore
the base-two expansion of $\CSb_V(T)$ with infinitely many ones has
infinitely many zeros also.
Hence $\alpha_n=\rest{\CSb_V(T)}{n}$ for every $n\in\N^+$.
It follows from \eqref{nHaob} that $\CSb_V(T)$ is weakly Chaitin $T$-random.

(ii) Suppose that $T$ is a computable real and $0<T<1$.
Note that
\begin{align*}
  \sum_{H_V(s)<T\abs{s}}(2^{-\abs{s}})^T
  &=\sum_{H_V(s)<T\abs{s}}2^{-T\abs{s}}
  <\sum_{H_V(s)<T\abs{s}}2^{-H_V(s)}\\
  &\le\sum_{s\in\X}2^{-H_V(s)}
  \le\sum_{p\in\Dom V}2^{-\abs{p}}
  =\Omega_V
  <\infty.
\end{align*}
Thus, since $T$ is computable, it is easy to show that
$\CSb_V(T)$ is a $T$-convergent left-computable real.
It follows from Theorem~\ref{mfcs10} that
$\CSb_V(T)$ is strictly $T$-compressible.

(iii)
Suppose that $T>1$.
Using \eqref{Hlabs},
it is easy to show that
there exists $n_0\in\N$ such that,
for every $s\in\X$,
if $\abs{s}\ge n_0$ then $H_V(s)<T\abs{s}$.
Thus, obviously, $\CSb_V(T)$ diverges to $\infty$.
\qed
\end{proof}

\vspace*{-4mm}

\section{Fixed Point Theorem on Partial Randomness by $\CS_V(T)$}
\label{fptgen1}

\vspace*{-1mm}

In this section,
we prove the following form of
fixed point theorem on partial randomness,
which is based on the computability of the value $\CS_V(T)$.
Note that this theorem has the same form as Theorem~\ref{fpt}.

\begin{theorem}[\boldmath
fixed point theorem on partial randomness
by $\CS_V(T)$]\label{main}
Let $V$ be an optimal prefix-free machine.
For every $T\in(0,1)$,
if $\CS_V(T)$ is computable,
then $T$ is weakly Chaitin $T$-random and $T$-compressible.
\qed
\end{theorem}


Let $V$ be an arbitrary optimal prefix-free machine in what follows.
Theorem~\ref{main} follows immediately from
Theorem~\ref{fpwcTr}, Theorem~\ref{fpTc1}, and Theorem~\ref{fpTc2} below,
as well as from Theorem~\ref{rfcomputable}~(i).
Let $s_1,s_2,s_3, \dotsc$ be
a particular recursive enumeration of
the infinite r.e.~set $\{\,s\in\X\mid H_V(s)<\abs{s}\,\}$.
For each $k\in\N^+$ and each real $x>0$,
we define $Z_k(x)$
as $\sum_{i=1}^{k} 2^{-\abs{s_i}/x}$.
Note then that $\lim_{k\to\infty}Z_k(x)=\CS_V(x)$ for every $x\in(0,1]$.

\begin{theorem}\label{fpwcTr}
For every $T\in(0,1)$,
if $\CS_V(T)$ is right-computable then $T$ is weakly Chaitin $T$-random.
\end{theorem}

\begin{proof}
First,
we define $W_k(x)$
as $\sum_{i=1}^{k} \abs{s_i}2^{-\abs{s_i}/x}$
for each $k\in\N^+$ and each real $x>0$.
We show that,
for each $x\in(0,1)$,
$W_k(x)$ converges
as $k\to\infty$.
Let $x$ be an arbitrary real with $x\in(0,1)$.
Since $x<1$,
there is $l_0\in\N^+$ such that
$(\log_2 l)/l\le 1/x-1$
for all $l\ge l_0$.
Then there is $k_0\in\N^+$ such that $\abs{s_i}\ge l_0$ for all $i>k_0$.
Thus, we see that,
for each $i>k_0$,
\begin{equation*}
  \abs{s_i}2^{-\frac{\abs{s_i}}{x}}
  =2^{-(\frac{1}{x}-\frac{\log_2 \abs{s_i}}{\abs{s_i}})\abs{s_i}}
  \le 2^{-\abs{s_i}}.
\end{equation*}
Hence, for each $k>k_0$,
$W_{k}(x)-W_{k_0}(x)
=\sum_{i=k_0+1}^{k} \abs{s_i}2^{-\abs{s_i}/x}
\le \sum_{i=k_0+1}^{k} 2^{-\abs{s_i}}<\CS_V$.
Therefore,
since $\{W_{k}(x)\}_k$ is
an increasing sequence of reals bounded to the above,
it
converges
as $k\to\infty$,
as desired.
For each $x\in(0,1)$,
we define a positive real $W(x)$ as
$\lim_{k\to\infty}W_k(x)$.

On the other hand,
since $\CS_V(T)$ is right-computable by the assumption,
there exists a total recursive function $f\colon\N^+\to\Q$ such that
$\CS_V(T)\le f(m)$ for all $m\in\N^+$, and
$\lim_{m\to\infty} f(m)=\CS_V(T)$.

We choose a particular real $t$ with $T<t<1$.
Then,
for each $i\in\N^+$, using the mean value theorem we see that
\begin{equation*}
  2^{-\frac{\abs{s_i}}{x}}-2^{-\frac{\abs{s_i}}{T}}
  < \frac{\ln 2}{T^2}\abs{s_i}2^{-\frac{\abs{s_i}}{t}}(x-T)
\end{equation*}
for all $x\in(T,t)$.
We
then
choose a particular $c\in\N$ with $W(t)\ln 2/T^2\le 2^c$.
Here,
the limit value $W(t)$ exists,
since $0<t<1$.
It follows that
\begin{equation}\label{self-contained}
  Z_k(x)-Z_k(T)<2^c(x-T)
\end{equation}
for all $k\in\N^+$ and $x\in(T,t)$.
We also choose a particular $n_0\in\N^+$ such that
$0.(\rest{T}{n})+2^{-n}<t$ for all $n\ge n_0$.
Such $n_0$ exists since $T<t$ and $\lim_{n\to\infty} 0.(\rest{T}{n})+2^{-n}=T$.
Since $\rest{T}{n}$ is the first $n$ bits of
the base-two expansion of $T$ with infinitely many zeros,
we then see that $T<0.(\rest{T}{n})+2^{-n}<t$ for all $n\ge n_0$.
In addition,
we choose a particular $n_1\in\N^+$ such that
$(n-c)2^{-n}\le 1$ for all $n\ge n_1$.
For each $n\ge 1$,
since $\abs{T-0.(\rest{T}{n})}<2^{-n}$,
we see that that
$\abs{T(n-c)-0.(\rest{T}{n})(n-c)}<(n-c)2^{-n}\le 1$.
Hence, we have
\begin{equation}\label{Tnn-cTn-c}
  \lfloor 0.(\rest{T}{n})(n-c)\rfloor\le T(n-c)\quad\&\quad
  T(n-c)-2\le \lfloor 0.(\rest{T}{n})(n-c)\rfloor
\end{equation}
for every $n\ge n_1$.
We define $n_2=\max\{n_0,n_1,c+1\}$.

Now,
given $\rest{T}{n}$ with $n\ge n_2$,
one can find $k_0,m_0\in\N^+$ such that
$f(m_0)<Z_{k_0}(0.(\rest{T}{n})+2^{-n})$.
This is possible
from $Z(T)<Z(0.(\rest{T}{n})+2^{-n})$,
\begin{equation*}
  \lim_{k\to\infty}Z_k(0.(\rest{T}{n})+2^{-n})=Z(0.(\rest{T}{n})+2^{-n}),
\end{equation*}
and the properties of $f$.
It follows from $Z(T)\le f(m_0)$
and \eqref{self-contained} that
\begin{equation*}
  \sum_{i=k_0+1}^{\infty} 2^{-\abs{s_i}/T}
  =Z(T)-Z_{k_0}(T)
  <Z_{k_0}(0.(\rest{T}{n})+2^{-n})-Z_{k_0}(T)<2^{c-n}.
\end{equation*}
Hence,
for every $i>k_0$,
$2^{-\abs{s_{i}}/T}<2^{c-n}$
and therefore $T(n-c)<\abs{s_i}$.
Thus,
by calculating the set
$\{\,s_i\mid i\le k_0\;\&\;\abs{s_i}=\lfloor 0.(\rest{T}{n})(n-c)\rfloor\,\}$
and picking any
one
finite binary string of length $\lfloor 0.(\rest{T}{n})(n-c)\rfloor$ which is not in this set,
one can obtain $s\in\{0,1\}^{\lfloor 0.(\reste{T}{n})(n-c)\rfloor}$ such that $\abs{s}\le H_V(s)$.
This is possible since
$\{\,s_i\mid i\le k_0\;\&\;\abs{s_i}=\lfloor 0.(\rest{T}{n})(n-c)\rfloor,\}
=S_V(\lfloor 0.(\rest{T}{n})(n-c)\rfloor)\subsetneqq \{0,1\}^{\lfloor 0.(\reste{T}{n})(n-c)\rfloor}$,
where the first equality follows from the first inequality in \eqref{Tnn-cTn-c} and
the last proper inclusion is due to the first half of Theorem~\ref{evalsHsln}.

Hence, there exists a partial recursive function $\Psi\colon\X\to\X$
such that
$\lfloor 0.(\rest{T}{n})(n-c)\rfloor\le H(\Psi(\rest{T}{n}))$
for all $n\ge n_2$.
Using \eqref{Psi}, there is $c_\Psi\in\N$ such that
$H(\Psi(\rest{T}{n}))\le H(\rest{T}{n})+c_\Psi$
for all $n\ge n_2$.
Thus,
it follows from the second inequality in \eqref{Tnn-cTn-c}
that
$Tn-Tc-2-c_\Psi<H(\rest{T}{n})$
for all $n\ge n_2$,
which implies that $T$ is weakly Chaitin $T$-random.
\qed
\end{proof}


\begin{theorem}\label{fpTc1}
For every $T\in(0,1)$,
if $\CS_V(T)$ is right-computable,
then $T$ is also right-computable.
\end{theorem}

\begin{proof}
Since $\CS_V(T)$ is right-computable,
there exists a total recursive function $f\colon\N^+\to\Q$ such that
$\CS_V(T)\le f(m)$ for all $m\in\N^+$, and
$\lim_{m\to\infty} f(m)=\CS_V(T)$.
Thus,
since $\CS_V(x)$ is an increasing function of $x\in(0,1]$,
we see that,
for every $x\in\Q$ with $0<x<1$,
$T<x$ if and only if there are $m,k\in\N^+$ such that
$f(m)<Z_k(x)$.
It follows from Theorem \ref{rfcomputable} (ii) that
$T$ is right-computable.
\qed
\end{proof}


\begin{theorem}\label{fpTc2}
For every $T\in(0,1)$,
if $\CS_V(T)$ is left-computable and $T$ is right-computable,
then $T$ is $T$-compressible.
\end{theorem}

\begin{proof}
%
For each $i\in\N^+$, using the mean value theorem we see that
\vspace*{-1mm}
\begin{equation*}
  2^{-\frac{\abs{s_1}}{t}}-2^{-\frac{\abs{s_1}}{T}}
  > (\ln 2)\abs{s_1}2^{-\frac{\abs{s_1}}{T}}(t-T)
\end{equation*}
\vspace*{-5mm}\\
for all $t\in(T,1)$.
We choose a particular $c\in\N^+$ such that
$(\ln 2)\abs{s_1}2^{-\frac{\abs{s_1}}{T}}\ge 2^{-c}$.
Then, it follows that
\vspace*{0mm}
\begin{equation}\label{self-contained2}
  Z_k(t)-Z_k(T)>2^{-c}(t-T)
\end{equation}
for all $k\in\N^+$ and $t\in(T,1)$.
\vspace*{-5mm}\\

Since $T$ is a right-computable real with $T<1$ by the assumption,
there exists a total recursive function $f\colon\N^+\to\Q$ such that
$T<f(l)< 1$ for all $l\in\N^+$, and
$\lim_{l\to\infty} f(l)=T$.
On the other hand,
since $\CS_V(T)$ is left-computable by the assumption,
there exists a total recursive function $g\colon\N^+\to\Q$ such that
$g(m)\le \CS_V(T)$ for all $m\in\N^+$, and
$\lim_{m\to\infty} g(m)=\CS_V(T)$.
By Theorem~\ref{CSVwCr}, $\CS_V$ is weakly Chaitin random
and therefore $\CS_V\notin\Q$.
Thus,
the base-two expansion of $\CS_V$ is unique and contains infinitely many ones,
and $0<\CS_V<1$ in particular.

Given $n$ and $\rest{\CS_V}{\lceil Tn\rceil}$
(i.e., the first $\lceil Tn\rceil$ bits of the base-two expansion of $\CS_V$),
one can find $k_0\in\N^+$ such that
$0.(\rest{\CS_V}{\lceil Tn\rceil})<
\sum_{i=1}^{k_0} 2^{-\abs{s_i}}$.
This is possible since
$0.(\rest{\CS_V}{\lceil Tn\rceil})<\CS_V$
and $\lim_{k\to\infty}\sum_{i=1}^{k} 2^{-\abs{s_i}}=\CS_V$.
It is then easy to see that
$\sum_{i=k_0+1}^{\infty} 2^{-\abs{s_i}}
=\CS_V-\sum_{i=1}^{k_0} 2^{-\abs{s_i}}
<2^{-\lceil Tn\rceil}\le 2^{-Tn}$.
Using the inequality $a^d+b^d\le (a+b)^d$
for any reals $a,b>0$ and $d\ge 1$,
it follows that
\vspace*{-1mm}
\begin{equation}\label{tcnew}
  \CS_V(T)-Z_{k_0}(T)
  =\sum_{i=k_0+1}^{\infty} 2^{-\frac{\abs{s_i}}{T}}
  <2^{-n}.
\end{equation}
\vspace*{-3mm}\\
Note that
$\lim_{l\to\infty} Z_{k_0}(f(l))=Z_{k_0}(T)$.
Thus, since $Z_{k_0}(T)<\CS_V(T)$,
one can then find $l_0, m_0\in\N^+$ such that
$Z_{k_0}(f(l_0))<g(m_0)$.
It follows from \eqref{tcnew} and \eqref{self-contained2}
that
$2^{-n}
>g(m_0)-Z_{k_0}(T)
>Z_{k_0}(f(l_0))-Z_{k_0}(T)
>2^{-c}(f(l_0)-T)$.
Thus, $0<f(l_0)-T<2^{c-n}$.
Let $t_n$ be the first $n$ bits of the base-two expansion of
the rational number $f(l_0)$ with infinitely many zeros.
Then, $\abs{\,f(l_0)-0.t_n\,}<2^{-n}$.
It follows from $\abs{\,T-0.(\rest{T}{n})\,}<2^{-n}$ that
$\abs{\,0.(\rest{T}{n})-0.t_n\,}<(2^c+2)2^{-n}$.
Hence, $\rest{T}{n}=t_n,\,t_n\pm 1,\,t_n\pm 2,\,\dots,\,t_n\pm (2^c+1)$,
where $\rest{T}{n}$ and $t_n$ are regarded as a dyadic integer.
Thus, there are still $2^{c+1}+3$ possibilities of $\rest{T}{n}$,
so that one needs only $c+2$ bits more in order to determine $\rest{T}{n}$.

Thus, there exists a partial recursive function
$\Phi\colon \N^+\times\X\times\X\to\X$ such that
\vspace*{-1mm}
\begin{equation}\label{PnCSVTnsTn}
  \forall\,n\in\N^+\quad\exists\,s\in\X\quad
  \abs{s}=c+2\;\;\&\;\;
  \Phi(n,\rest{\CS_V}{\lceil Tn\rceil},s)=\rest{T}{n}.
\end{equation}
\vspace*{-5mm}

Let us consider a prefix-free machine $D$ which
satisfies the following two conditions (i) and (ii):
(i) For each $p,q\in\Dom U$ and $v,s\in\X$,
$pqvs\in\Dom D$ if and only if $\abs{v}=U(q)$ and $\abs{s}=c+2$.
(ii) For each $p,q\in\Dom U$ and $v,s\in\X$ such that $\abs{v}=U(q)$ and $\abs{s}=c+2$,
$D(pqvs)=\Phi(U(p),v,s)$.
It is easy to see that such a prefix-free machine $D$ exists.
For each $n\in\N^+$,
note that $n=U(n^*)$ and $\abs{\rest{\CS_V}{\lceil Tn\rceil}}=U(\lceil Tn\rceil^*)$.
Thus, it follows from \eqref{PnCSVTnsTn} that
there exists $s\in\X$ with $\abs{s}=c+2$
such that
$D(n^*\lceil Tn\rceil^* \rest{\CS_V}{\lceil Tn\rceil} s)
=\Phi(n,\rest{\CS_V}{\lceil Tn\rceil},s)=\rest{T}{n}$.
Hence,
$H_D(\rest{T}{n})
\le \abs{n^*}+\abs{\lceil Tn\rceil^*}+\abs{\rest{\CS_V}{\lceil Tn\rceil}}+\abs{s}
=H(n)+H(\lceil Tn\rceil)+\lceil Tn\rceil+c+2$.
It follows from \eqref{Hlabs} that
$H_D(\rest{T}{n})\le Tn+2\log_2 n+2\log_2\log_2 n+O(1)$ for all $n\in\N^+$.
Using \eqref{minimal}
we see that $T$ is $T$-compressible.
\qed
\end{proof}

\vspace*{-4mm}

\subsubsection*{Acknowledgments.}

This work was supported
by KAKENHI, Grant-in-Aid for Scientific Research (C) (20540134),
by SCOPE
from the Ministry of Internal Affairs and Communications of Japan,
and by CREST from Japan Science and Technology Agency.

\vspace*{-2mm}


\end{document}